\documentclass[10pt, conference]{ieeeconf}

\usepackage{graphics} 
\usepackage{amsmath} 
\usepackage{amssymb}  
\usepackage{cite}
\usepackage{bm}
\usepackage{acronym}
\usepackage{paralist}
\usepackage{float}
\usepackage{color}
\usepackage{epstopdf}
\usepackage{multicol}
\usepackage{tikz}
\usepackage{hyperref}

\newtheorem{Theorem}{Theorem}

\newtheorem{Lemma}{Lemma}

\newtheorem{Problem}{Problem}

\newtheorem{Remark}{Remark}

\newtheorem{Assumption}{Assumption}

\newtheorem{Definition}{Definition}

\newcommand{\bequ}{\begin{eqnarray}}
\newcommand{\eequ}{\end{eqnarray}}

\def\IR{{\mathbb R}}

\newcommand{\ds}{\text{d}s}

\IEEEoverridecommandlockouts
\begin{document}

\title{\LARGE \bf Prescribed-time convergence with input constraints: A control Lyapunov function based approach}
\author{Kunal Garg \and Ehsan Arabi \and Dimitra Panagou
\thanks{The authors are with the Department of Aerospace Engineering, University of Michigan, Ann Arbor, MI, USA; \texttt{\{kgarg, earabi, dpanagou\}@umich.edu}.}
\thanks{The authors would like to acknowledge the support of the Air Force Office of Scientific
Research under award number FA9550-17-1-0284.}}

\maketitle

\begin{abstract}
    In this paper, we present a control framework for a general class of control-affine nonlinear systems under spatiotemporal and input constraints. Specifically, the proposed control architecture addresses the problem of reaching a given final set $S$ in a prescribed (user-defined) time with bounded control inputs. To this end, a time transformation technique is utilized to transform the system subject to temporal constraints into an equivalent form without temporal constraints. The transformation is defined so that asymptotic convergence in the transformed time scale results into prescribed-time convergence in the original time scale. To incorporate input constraints, we characterize a set of initial conditions $D_M$ such that starting from this set, the closed-loop trajectories reach the set $S$ within the prescribed time. We further show that starting from outside the set $D_M$, the system trajectories reach the set $D_M$ in a finite time that depends upon the initial conditions and the control input bounds. We use a novel parameter $\mu$ in the controller, that controls the convergence-rate of the closed-loop trajectories and dictates the size of the set $D_M$. Finally, we present a numerical example to showcase the efficacy of our proposed method. 
\end{abstract}

\section{Introduction}

In many real-world applications, various types of constraints are present due to the structural and operational requirements of the considered system. 
For example, spatial constraints are common in safety-critical applications. 
There are a lot of studies on \textit{forward invariance} of sets where the objective is to design a control law such that the closed-loop system trajectories are contained in a given set for all times.
Some examples include: \cite{tee2009barrier}, where a Lyapunov-like barrier function based approach is utilized to guarantee asymptotic tracking, as well as ensuring that the system output always remains inside a given set; \cite{ames2017control,ames2014control}, where conditions using zeroing barrier functions are presented to ensure forward invariance of a desired set. The authors of \cite{barry2012safety} present sufficient conditions in terms of existence of a barrier certificate for forward invariance of a given set, and propose a sum-of-squares formulation to find the certificate. 

In addition to spatial constraints, temporal constraints appear in time-critical applications where completion of a task is required upon a given time instance. The concept of finite-time stability (FTS) has been studied to ensure convergence of solutions in finite time. In the seminal work \cite{bhat2000finite}, the authors introduce necessary and sufficient conditions in terms of a Lyapunov function for continuous, autonomous systems to exhibit FTS. Fixed-time stability (FxTS) \cite{polyakov2012nonlinear} is a stronger notion than FTS, where the time of convergence does not depend upon the initial conditions. Prescribed-time stability, or user-defined time stability or strongly predefined-time stability \cite{kan2017finite,aldana2019design,holloway2019prescribed}, imposes that the convergence time can be chosen arbitrarily. The authors of \cite{kan2017finite,yucelen2018finite} study a time transformation approach to \textit{stretch} the finite-time domain of interest $[0,T)$ to an infinite-time domain $ [0,\infty)$. With a proper choice of the transformation, the asymptotic convergence in the new stretched time domain inherently implies the finite-time convergence at the prescribed time $T$ in the original time domain (see also \cite{arabi2018further,arabi2019robustness,arabi2019spatio,song2017time} for more details). Input constraints, such as actuator saturation, is another class of constraints that is inevitable in practice. Since a limited control input can affect the region of fixed-time convergence, addressing spatiotemporal and input constraints simultaneously is a challenging control problem. 

In contrast to the forward invariance problem, the problem of reaching a user-defined set deals with designing a control law such that the closed-loop system trajectories, starting from outside some set $S$, reach the set $S$ in a given time. When the set $S$ contains only the equilibrium, then reaching the set $S$ reduces to the regular point stabilization problem. 
Control synthesis for the problem of reaching a general set $S$ has received much attention. A recent result in  \cite{li2019finite} introduces the notion of FTS of closed sets for hybrid dynamical systems. The authors of \cite{li2018formally} formulate a quadratic program to ensure finite-time convergence to a set with input constraints. Although they consider input bounds, the approach does not yield closed-form controllers. Under the traditional notion of FTS, as defined in \cite{bhat2000finite}, the convergence time in \cite{li2019finite,li2018formally} depends upon the initial conditions.

In this paper, we study the problem of reaching a given set in a prescribed-time $T$ for a general class of control-affine systems with input constraints. We present sufficient conditions for reaching a general set $S = \{x\; |\; h(x)\leq 0\}$ in time $T$. We present novel sufficient conditions in terms of existence of a CLF to guarantee existence of a controller that derives the closed-loop trajectories to the set $S$. 

In contrast to the results in \cite{li2018formally}, we propose a control law in closed-form. 
Furthermore, in contrast to the results in \cite{li2019finite}, the closed-loop system trajectories resulting from our controller reach the given set in a prescribed-time that can be chosen arbitrarily and independently of the initial conditions in the absence of input constraints. To incorporate input constraints, we then characterize a set of initial conditions $D_M$ starting from which, the closed-loop trajectories reach the set $S$ in the prescribed time $T$. As expected, assigning an arbitrary prescribed convergence time for any initial condition may not be possible in the presence of actuator saturation. To this end, we derive a lower bound on the prescribed time of convergence $T$, so that the control input remains bounded at all times. When the initial conditions of the system are outside the set $D_M$, we show that the proposed controller yields convergence to the set $D_M$ within a finite time that depends on the initial conditions and the control input bounds. We use a novel tuning parameter $\mu$ in the controller and show how the size of the set $D_M$ and the convergence rate can be adjusted based on the value of this parameter. 

The rest of the paper is organized as follows. Section \ref{sec: math prelim} presents the notation used in the paper and some mathematical preliminaries. In Section \ref{sec: cont des}, we detail the controller design and present the main results of the paper. In Section \ref{sec: simulation}, we present a numerical example to explain the various aspects of the proposed method and we conclude the paper with our thoughts for future work in Section \ref{sec: conclusion}. 
    
\section{Mathematical Preliminaries}\label{sec: math prelim}
In the rest of the paper, $\IR$ denotes the set of real numbers, $\mathbb R_+$ denotes the set of non-negative real numbers,  $\IR^n$ denotes the set of $n \times 1$ real column vectors,
and $f(T^-)$ (respectively, $f(T^+)$)  denotes the left (respectively, the right) limit of the function $f$ at $T$.
In addition, we write $\textrm{bd}(S)$ for the boundary of the closed set $S$,
and $\|x\|_S = \inf_{y\in S}\|x-y\|$ for the distance of the point $x\notin S$ from the set $S$. A function $\alpha:\mathbb R_+\rightarrow \mathbb R_+$ is said to belong to class-$\mathcal K$, denoted as $\alpha\in \mathcal K$, if $\alpha(0) = 0$ and $\alpha$ is strictly increasing. The Lie derivative of a function $V:\mathbb R^n\rightarrow \mathbb R$ along a vector field $f:\mathbb R^n\rightarrow\mathbb R^n$ at a point $x\in \mathbb R^n$ is denoted as $L_fV \triangleq \frac{\partial V}{\partial x} f(x)$. 

Next, we introduce the notion of prescribed-time stability. Consider a nonlinear system given by 
\begin{align}\label{ex sys}
\dot x(t) = f(x(t)), \quad x(0) = x_0,
\end{align}
where $x\in \mathbb R^n$ and $f: \mathbb R^n \rightarrow \mathbb R^n$ is continuous with $f(0)=0$. The origin is said to be an FTS equilibrium of \eqref{ex sys} if it is Lyapunov stable and \textit{finite-time convergent}, i.e., for all $x(0) \in \mathcal N \setminus\{0\}$, where $\mathcal N$ is some open neighborhood of the origin, $\lim_{t\to T} x(t)=0$, where $T = T(x(0))<\infty$ depends upon the initial condition $x(0)$ \cite{bhat2000finite}. The authors in \cite{polyakov2012nonlinear} define the notion of FxTS, where the time of convergence does not depend upon the initial condition. The following notion of prescribed-time stability allows the time of convergence $T$ to be chosen a priori. 
\begin{Definition}[\cite{kan2017finite,aldana2019design,holloway2019prescribed}]
The origin of \eqref{ex sys} is called as prescribed-time stable if the trajectories of \eqref{ex sys} reach the origin in time $T<\infty$, where $T>0$ is a user defined constant.
\end{Definition}

We introduce a time transformation technique for converting a user-defined time interval of interest $[0,T)$ into a stretched infinite-time interval $[0,\infty)$ \cite{kan2017finite,yucelen2018finite,arabi2018further}. Consider the function $\theta:\mathbb R_+\rightarrow \mathbb R_+$ defined as 
\begin{align}\label{t theta s}
    t = \theta (s), 
\end{align}
where $t\in [0, T)$, $s\in [0, \infty)$, and the function $\theta(s)$ satisfies 
\begin{subequations}\label{theta prop}
\begin{align}
    \theta(0)  & =0, \; \theta'(0) = 1,\label{theta 0 cond}\\
    s_1 >s_2\geq 0 &\implies \theta(s_1) > \theta(s_2),\label{theta inc cond}\\
    \lim_{s\to\infty}\theta(s) &= T,\; \lim_{s\to\infty}\theta'(s) = 0. \label{theta prop_d}
\end{align}
\end{subequations}
We need \eqref{theta 0 cond} to ensure continuity of the proposed controller, \eqref{theta inc cond} implies that $\theta(s)$ is strictly increasing, while \eqref{theta prop_d} is needed to ensure that asymptotic convergence in $s$ results into prescribed time convergence in $t$. A candidate time transformation function satisfying \eqref{theta prop} is given by 
\bequ
t = \theta(s) \triangleq T(1-e^{-\frac{s}{T}}),  \label{P1_TT}
\eequ
where $T>0$ is the user-defined finite time. 

\begin{Remark}[\cite{benner2014large} Section 1.1.1.4] \label{rem_benner}
Let $\xi(t)$ denote a solution to the dynamical system
\bequ
\dot{x}(t) = f(t,x(t)), \quad x(0) = x_0.
\eequ
Let $t=\theta(s)$ satisfy \eqref{theta prop}, and define $\chi(s) = \xi(t)$ so that
\bequ
\chi'(s) = \theta'(s) f(\theta(s),\chi(s)), \quad \chi(\theta^{-1}(0)) = x_0,
\eequ
where  $\chi'(s) \triangleq \text{d}\chi(s)/\ds$, $\theta'(s)\triangleq \text{d}\theta(s)/\ds$ and $\lim_{s\to\infty}\chi(s) = \lim_{t\to T}\xi(t)$. 
\end{Remark} 

\begin{Remark}  \label{rem_xs}
For the sake of simplicity, considering the time transformation $t=\theta(s)$ and any signal $\eta(t)$, we write $\eta_s(s)$ to denote the signal $\eta(t)$ in the transformed time coordinate $s$, i.e., $\eta_s(s) \triangleq \eta(\theta(s))$. 
\end{Remark} 

\section{Prescribed-time Control design}\label{sec: cont des}
Consider the system
\begin{align}\label{cont affin sys}
    \dot x(t) = f(x) + g(x)u(t), \quad x(0) =x_0,
\end{align}
where $x(t)\in\mathbb R^n$ is the system state vector, $f:\mathbb R^n\rightarrow\mathbb R^n$, $g:\mathbb R^n\rightarrow\mathbb R^n$ and $u\in\mathbb R$ is the control input.
The problem statement is as follows. 
\begin{Problem}\label{P1}
Design the control input $u(t)$ so that 
\begin{itemize}
    \item The closed-loop trajectories of \eqref{cont affin sys} reach the set $S = \{x\; |\; h(x)\leq 0\}$ for a given function $h:\mathbb R^n\rightarrow \mathbb R$ in a user-defined time $T$, i.e., $\forall t\geq T$ and $\forall x_0 \in D\subset \mathbb R^n$, $x(t)\in S$;
    \item For a user-defined control bound $u_M>0$, $u$ satisfies
    \begin{align}\label{um bound P1}
        -u_M \leq u(t) \leq u_M, \; \forall \; t\geq 0.
    \end{align}
\end{itemize}
\end{Problem}

First, we utilize the time transformation given in \eqref{t theta s} to represent the system dynamics in \eqref{cont affin sys} in the new stretched infinite time interval $[0,\infty)$. From Remarks \ref{rem_benner} and \ref{rem_xs}, we obtain
\begin{align}\label{cont affin sys s}
    x_s'(s) = \bar f(s,x_s(s)) + \bar g(s,x_s(s))u_s(s),
\end{align}
where $x_s(s) = x(t)$, $x_s(\theta^{-1}(0)) = x_0$, $\bar f(s,x) = f(x)\theta'(s)$ and $\bar g(s,x) = g(x)\theta'(s)$. We now re-state Problem \ref{P1} in the new time-scale for \eqref{cont affin sys s}.

\begin{Problem}\label{P1 mod}
Design the control input $u_s(s)$ such that 
\begin{itemize}
    \item The closed-loop trajectories of \eqref{cont affin sys s} reach the set $S = \{x\; |\; h(x)\leq 0\}$ for a given function $h:\mathbb R^n \rightarrow \mathbb R$ asymptotically, i.e., $\forall x_s(0)\in D\subset \mathbb R^n$, $\lim_{s\to\infty} x_s(s)\in S$;
    \item The control input satisfies 
    \begin{align}\label{u bound s P2}
        -u_M \leq u_s(s) \leq u_M, \; \forall \; s\geq 0,
    \end{align}
    for a user-defined control bound $u_M>0$.
\end{itemize}
\end{Problem}

The control Lyapunov function (CLF) is defined as follows. 
\begin{Definition}[\cite{sontag1989universal}]\label{CLF def}
A continuously differentiable, positive definite function $V:\mathbb R^n\rightarrow\mathbb R$ is called CLF for \eqref{cont affin sys} if for all $t\geq 0$ it satisfies
\begin{align}
    \inf\limits_{u}\{L_fV(x(t)) + L_gV(x(t))u\} \leq 0,
\end{align}
\end{Definition}


We make the following assumption for system \eqref{cont affin sys}. 
\begin{Assumption}\label{assum CLF exist}
There exists a CLF $V$ for \eqref{cont affin sys}.
\end{Assumption}

\begin{Remark}
In contrast to \cite{li2018formally}, where the authors use the function $h(x)$ as the CLF, Assumption \ref{assum CLF exist} is more general and gives more freedom to choose the CLF as some function other than $h(x)$. It also allows more general class of set $S$ that can be considered since we do not need the function $h(x)$ to be continuously differentiable. 
\end{Remark}

From Definition \ref{CLF def}, $V$ is positive definite. Using \cite[Lemma 4.3]{khalil2002nonlinear}, there exist $\alpha_1, \alpha_2\in \mathcal K$ such that:
\begin{align}\label{v bound}
    \alpha_1(\|x\|_S) \leq V(x)\leq \alpha_2(\|x\|_S) ,\; \forall x.
\end{align}

Now, we present the controller for \eqref{cont affin sys s} so that its trajectories reach the set $S$ asymptotically. Define the functions $a, b$ as 
\begin{align}\label{a b func}
    a(s,x_s) &= L_{\bar f}V(x_s),\; b(s,x_s) = L_{\bar g}V(x_s).
\end{align}
Also, define the functions $a_0, b_0$ as
\begin{align}\label{a0 b0 func}
    a_0(x_s) &= L_{f}V(x_s),\; b_0(x_s) = L_{g}V(x_s).
\end{align}
Note that for time-invariant functions $f,g$, the functions $a_0, b_0$ are also time invariant. From \eqref{a b func} and \eqref{a0 b0 func}, we have
\begin{align}\label{a b a0 b0 rel}
    a_0(x_s) &= \frac{1}{\theta'(s)}a(s,x_s),\;    b_0(x_s) =\frac{1}{\theta'(s)}b(s,x_s).
\end{align}
Using \eqref{a b func} and inspired from Sontag's formula \cite{sontag1989universal}, we define the control signal $\bar u(\cdot)$ as
\begin{align}\label{u s}
   \bar u(\theta'(s),x_s(s)) = \begin{cases}
    -\frac{a + \frac{\mu}{\theta'(s)}\sqrt{a^2+\frac{b^4}{(\theta'(s))^2}}}{b}, & b\neq 0, \\
    0, & b = 0,
    \end{cases}
\end{align}
where $\mu>0$ is a tuning parameter (see Remark \ref{rem7} for details on how to choose this parameter). In order to incorporate input constraints \eqref{u bound s P2}, define the set $D_0$ as
\begin{align}\label{set D0}
    D_0 = \{x\; |\; |\bar u(\cdot, \cdot)|\leq u_M\}.
\end{align}

\begin{Remark} 
The system \eqref{cont affin sys s} is a time-varying system, hence, a stabilizing feedback law for time-invariant system \eqref{cont affin sys} does not necessarily stabilize \eqref{cont affin sys s} (see \cite{jiang2009stabilization} for details).
\end{Remark}


\begin{Lemma}\label{V prime bound}
Along the closed-loop trajectories of \eqref{cont affin sys s} with $u_s = \bar u$, the CLF $V$ satisfies:
\begin{align}\label{v prime bound}
    V'(x_s) \triangleq \frac{\text{d}V(x_s)}{\ds} = -\mu\sqrt{a_0(x_s)^2+b_0(x_s)^4},
\end{align}
for all $x_s$ such that $b(s,x_s(s)) \neq 0$. 
\end{Lemma}
\begin{proof}
For $x_s$ such that $b\neq 0$, the control input is given as $\bar u = -\big({a + \frac{\mu}{\theta'(s)}\sqrt{a^2+\frac{b^4}{(\theta'(s))^2}}}\big)/{b}$. The derivative of $V$ along the closed-loop trajectories of \eqref{cont affin sys s} results in
\begin{align*}
    V' &= L_{\bar f}V +L_{\bar g}Vu= a+b\left(-\frac{a + \frac{\mu}{\theta'(s)}\sqrt{a^2+\frac{b^4}{(\theta'(s))^2}}}{b}\right), \\
    & = -\mu\sqrt{\Big(\frac{a}{\theta'(s)}\Big)^2+\Big(\frac{b}{(\theta'(s))}\Big)^4} \overset{\eqref{a b a0 b0 rel}}{=}-\mu\sqrt{a_0^2+b_0^4}.
\end{align*}
\end{proof}

We make the following assumptions. 
\begin{Assumption}\label{assum: l k1 k2}
There exist positive constants $k_1,k_2,l$ such that the following hold for all $x\in D_1\subset \mathbb R^n$:
\begin{subequations} \label{a0 b0 bounds}
\begin{align}
    \|a_0(x)\| &= \Big\|\frac{\partial V}{\partial x}f(x)\Big\| \leq l\|x\|_S^2, \\
    k_1\|x\|_S\leq \|b_0(x)\| &= \Big\|\frac{\partial V}{\partial x}g(x)\Big\| \leq k_2\|x\|_S ,\label{b0 inequality}
\end{align}
\end{subequations}
where $D_1$ is an open set containing $S$.
\end{Assumption}

Note that \eqref{b0 inequality} implies that  
\begin{align}\label{w bound}
    W(x) \triangleq \sqrt{a_0(x)^2+b_0(x)^4} \geq c\|x\|_S^2,
\end{align}
where $c = k_1^2$, for all $x\in D_1\subset\mathbb R^n$.

\begin{Assumption}\label{assum W c}
There exist $c_1, c_2>0$ such that 
\begin{subequations}\label{c1 c2 bound}
\begin{align}
    c_1\|x\|_S^2 \leq \alpha_1(\|x\|_S), \label{c1 bound}\\
    \alpha_2(\|x\|_S) \leq c_2\|x\|_S^2, \label{c2 bound}
\end{align}
\end{subequations}
for all $x\in D_2\subset \mathbb R^n$, $D_2$ is an open set containing $S$.
\end{Assumption}

\begin{Remark}
Similar assumptions have been used in literature previously. Authors in \cite{ames2014rapidly} use these assumptions to define an exponentially stabilizing CLF (ES-CLF); \cite{romdlony2016stabilization} uses \eqref{c1 c2 bound} to show asymptotic stability; \cite{jankovic2001control} uses \eqref{w bound} in the definition of CLF (see \cite[Remark 5]{pepe2013sontag} as well).
\end{Remark}

We state the following results on boundedness and continuity of $\bar u$.
\begin{Lemma} \label{lemma u bound}
The following holds for some $k>0$, for all $s\geq 0$
\begin{align}\label{u bound k V0}
    |\bar u(\theta'(s),x_s(s))| \leq  k\sqrt{\frac{V_0}{c_1}}e^{(\frac{1}{T}-\frac{\mu c}{2c_2})s}.
\end{align}
Additionally, if $T>\frac{2c_2}{\mu c}$, then $\lim_{s\to\infty}\bar u = 0$.  
\end{Lemma}
\begin{proof}
See Appendix \ref{App proof lemma 1}. 
\end{proof}

In what follows, we assume that $ T>\frac{2c_2}{\mu c}$. Define
\begin{align}\label{set Dm}
    D_M = \{x\; |\; V(x)\leq c_1\left(\frac{u_M}{k}\right)^2\},
\end{align}
and $m_2 = \frac{\mu c}{2 c_2}$, so that for $x_0\in D_M$, we have 
\begin{align*}
    |\bar u|& \leq k\sqrt{\frac{V_0}{c_1}}e^{-m_2s+s/T}\leq k\sqrt{\frac{V_0}{c_1}}\leq  u_M.
\end{align*}
Note that $D_M\subseteq D_0$ where $D_0$ is defined as in \eqref{set D0}. We propose the control input as
\begin{align}\label{u t act}
    u(t) = \begin{cases}
    \bar u(\theta'(\theta^{-1}(t-T_M)),x(t)), & x(t)\in D_M,  \\
    \bar u(1,x(t)), & x(t)\in D_0\setminus D_M,  \\
    u_M\frac{\bar u(1,x(t))}{|\bar u(1,x(t))|}, & x(t)\notin D_0,\end{cases}
\end{align}
where $T_M$ is the time instant when the trajectories of \eqref{cont affin sys} enter the set $D_M$, i.e., $x(T_M)\in \textrm{bd}(D_M)$ and $x(t)\notin D_M$ for all $t<T_M$, and
\begin{align}\label{u t}
   \bar u(y,x) = \begin{cases}
    -\frac{a_0(x) + \frac{\mu}{y}\sqrt{a_0(x)^2+b_0(x)^4}}{b_0(x)}, & b_0(x)\neq 0, \\
    0, & b_0(x) = 0.
    \end{cases}
\end{align}
From \eqref{u t}, we obtain that 
\begin{align*}
    \bar u(1,x) = -\frac{a_0(x) +\mu\sqrt{a_0(x)^2+b_0(x)^4}}{b_0(x)},
\end{align*}
which is time-invariant and depends only on $x$. Note that the time transformation technique is utilized only within the set $D_M$,  i.e., $t=T_M + \theta(s)$ and we have $s = 0$ when the system trajectories of \eqref{cont affin sys} enter the set $D_M$ at $t = T_M$.

\begin{Lemma}\label{lemmma u cont}
The control input $u(t)$ defined as \eqref{u t act} is continuous and satisfies the control input constraints \eqref{um bound P1}. 
\end{Lemma}
\begin{proof}
See Appendix \ref{App proof lemma 2}.
\end{proof}

We divide the problem in two parts: 
\begin{itemize}
    \item In Theorem \ref{T conv origin}, we show that for $x(0)\in D_M$,  the closed loop trajectories reach the set $S$ in time $T$.
    \item In Theorem \ref{finite time set reach}, we show that if $x(0)\notin D_M$, then the closed loop trajectories of \eqref{cont affin sys} reach the set $D_M$ in a finite time $T_0(x_0) + T_1$. 
\end{itemize}

Define $D = D_1\bigcap D_2$, where $D_1$ is defined in Assumption \ref{assum: l k1 k2}, and $D_2$ is defined in Assumption \ref{assum W c}.

\begin{Theorem}\label{T conv origin}
Under the effect of control input \eqref{u t act}, the trajectories of \eqref{cont affin sys} reach the set $S$ in the prescribed time $T$ for all $x_s(0)\in D_M\bigcap D$.  
\end{Theorem}
\begin{proof}
Based on \eqref{v bound}, $V(x_s) = 0$ implies $\|x_s\|_S = 0$. So, it is sufficient to prove that $V(x_s(s)) = 0$ as $s\rightarrow 0$. Now, for $x_s\in \textrm{bd}(D_M)$, it follows from \eqref{v prime bound}, \eqref{w bound} and \eqref{c1 c2 bound} that $V'(x_s) \leq -\frac{c}{c_2}V(x_s)<0$. Hence, we obtain that the set $D_M$ is forward-invariant. This implies that $x_s(s) \in  D_M$ for all $s\geq 0$. From \eqref{u t act}, we have $u(t) = \bar u(\theta'(s),x_s)$ for $x_s\in D_M$ and hence, $V'(x_s) \leq -\frac{c}{c_2}V(x_s)$ holds for all $s\geq 0$. Therefore, $\lim_{s\to\infty}V(x_s(s)) = 0$, which in turn implies that $\lim_{s\to\infty}\|x_s(s)\|_S = \lim_{t\rightarrow T}\|x(t)\|_S = 0$. 
\end{proof} 

\begin{Remark}
Based on \eqref{set Dm}, the set $D_M$ grows larger as the input bound $u_M$ increases. So, for larger input bounds, the set of initial conditions, from which the trajectories converge to the set $S$ in the prescribed time, is larger.
\end{Remark}

\begin{Theorem}\label{finite time set reach}
If  $x_0\notin D_M$, then the closed-loop trajectories of \eqref{cont affin sys} reach the set $D_M$ in a finite-time $T_0(x_0) + T_1$ for all $x_0\in \bar D \bigcap D$ where $\bar D= \{x\; |\; \|x\|_S \leq \frac{\mu cu_Mc_1}{lc_2k}\}$. If, in addition,  $x_0\in D_0 \setminus D_M$, then $T_0(x_0) = 0$. 
\end{Theorem}
\begin{proof}
Consider the case when $x_0\notin D_0$. From \eqref{u t act}, we have $u(t) =  u_M\frac{\bar u(1,x)}{|\bar u(1,x)|}$. The time derivative $\dot V(x(t))$ reads
\begin{align*}
    \dot V &= a_0 + b_0u = a_0 - b_0\frac{u_M}{|\bar u(1,x)|} \frac{a_0 +\mu\sqrt{a_0^2+b_0^4}}{b_0} \\
    & = a_0 - \frac{u_M}{|\bar u(1,x)|} (a_0 +\mu\sqrt{a_0^2+b_0^4})\\
    & = a_0(1-\frac{u_M}{|\bar u(1,x)|}) -\frac{u_M\mu}{|\bar u(1,x)|}\sqrt{a_0^2+b_0^4}.
\end{align*}
Now, if $a_0<0$, we obtain that  $a_0\big(1-\frac{u_M}{|\bar u(1,x)|}\big)<0$ since $|\bar u(1,x)|>u_M$. Using this, we obtain
\begin{align}\label{dot v ao neg}
    \dot V &\leq -\frac{u_M\mu}{|\bar u(1,x)|}\sqrt{a_0^2+b_0^4} \overset{\eqref{w bound}}{\leq} -c\mu \frac{u_M}{|\bar u(1,x)|}\|x\|_S^2  \nonumber \\
    &\overset{\eqref{u bound2}}{\leq} -c\mu\frac{u_M}{k\|x\|}\|x\|_S^2 = -\delta\mu\|x\|_S,
\end{align}
where $k = \big({l(1+\mu)}/{k_1}   + \mu k_2\big)$ and $\delta = {cu_M}/{k}$. Now, using \eqref{v bound}, \eqref{c1 c2 bound} and \eqref{dot v ao neg}, we obtain 
\begin{align}\label{dot V a0 neg ineq}
    \dot V\leq -\gamma V^{\frac{1}{2}},
\end{align}
where $\gamma = \mu\delta c_2^{-\frac{1}{2}}$. Define $\bar V = \min\limits_{x\in \textrm{bd}(D_0)} V(x)$. Now, since $x_0\notin D_0$, we know that $V(x(0))>\bar V$. Denote $t = \bar T$ such that $V(x(\bar T)) = \bar V$. From \eqref{dot V a0 neg ineq}, using Comparison lemma \cite[Lemma 3.4]{khalil2002nonlinear}, we obtain that 
\begin{align*}
    & V(0)^{\frac{1}{2}}-\bar V^\frac{1}{2} \geq \frac{1}{2}\gamma \bar T\implies \bar T\leq \frac{2(V(0)^{\frac{1}{2}}-\bar V^\frac{1}{2})}{\gamma}.
\end{align*}
For the case when $a_0>0$, we have
\begin{align*}
    \dot V &= a_0(1-\frac{u_M}{|\bar u(1,x)|}) -\frac{u_M\mu}{|\bar u(1,x)|}\sqrt{a_0^2+b_0^4}\\
    & \overset{\eqref{a0 b0 bounds}}{\leq }l\|x\|^2-\frac{u_M\mu}{|\bar u(1,x)|}\sqrt{a_0^2+b_0^4}\\
    &\overset{\eqref{dot V a0 neg ineq}}{\leq} l\|x\|^2-\gamma V^{\frac{1}{2}} \overset{\eqref{c1 c2 bound}}{\leq}\frac{l}{c_1}V-\gamma V^{\frac{1}{2}}.
\end{align*}
Now, for $x_0\in \bar D = \{x\; |\; \|x\|_S\leq \frac{\gamma c_1}{l\sqrt{c_2}}\}$, we know that $\frac{l}{c_1}V-\gamma V^{\frac{1}{2}}<0$. Using \cite[Lemma 1]{shen2008semi}, we obtain that for all $x_0\in \bar D$, there exists a finite time $T(x_0)$, such that $V(x(T)) = 0$, which implies that there exists a time $\tilde T(x_0)$ such that $V(x(\tilde T)) = \bar V$. Choose $T_0(x_0) = \max\{\tilde T(x_0), \bar T(x_0)\}$ so that for all $x_0\in \bar D$, the trajectories of \eqref{cont affin sys} reach the set $D_0$ in time $T_0(x_0)$. 

Once the trajectories reach $D_0$, the control input is defined as $u(t) = \bar u(1,x)$, so the inequality \eqref{xbound} holds. Let $x_{\max} = \max\limits_{x\in \textrm{bd}(D_0)}\|x\|_S$. Using \eqref{xbound}, we obtain
\begin{align}\label{x new bound}
    \|x(t+T_0)\|&\leq m_1 e^{-m_2 t} =  \sqrt{\frac{V(T_0)}{c_1}}e^{-m_2t} \nonumber\\
    &\overset{\eqref{c1 c2 bound}}{\leq}\sqrt{\frac{c_2\|x(T_0)\|^2}{c_1}}e^{-m_2t}\leq \sqrt{\frac{c_2}{c_1}}e^{-m_2t}x_{\max}.
\end{align}
Let $t = T_0 +T_1$ be the time instant when the trajectories of \eqref{cont affin sys} enter the set $D_M$, i.e., $x(T_0+T_1)\in \textrm{bd}(D_M)$ and $x(t)\notin D_M$ for all $t<T_0+T_1$. From \eqref{set Dm}, we know that $x\in \textrm{bd}(D_M) \implies \|x\|_S \leq \frac{u_M}{k}\Big(\frac{c_1}{c_2}\Big)^{\frac{1}{2}}$. Using this and \eqref{x new bound}, we obtain
\begin{align*}
    \frac{u_M}{k}\Big(\frac{c_1}{c_2}\Big)^{\frac{1}{2}} \leq \sqrt{\frac{c_2}{c_1}}e^{-m_2t}x_{max}
    \implies t \geq \frac{\log\Big(\frac{c_1}{c_2}\frac{u_M}{k}\Big)}{\log m_2},
\end{align*}
which implies that for all $t\geq  \frac{\log\Big(\frac{c_1}{c_2}\frac{u_M}{kx_{max}}\Big)}{\log m_2}$, $x(t)\in D_M$. Since the above analysis computes the worst case time to reach the set $D_M$, i.e., starting from $x_{max}$, we obtain that $T_1 \leq \frac{\log\Big(\frac{c_1}{c_2}\frac{u_M}{kx_{max}}\Big)}{\log m_2}$.
Hence, for any $x_0\in \bar D$, the closed-loop trajectories of \eqref{cont affin sys} reach the set $D_M$ in a finite-time $T_0(x_0) + T_1$. Since $T_0(x_0)$ is the time required to reach the set $D_0$, we obtain that if $x_0\in D_0$, then $T_0(x_0) = 0$. 
\end{proof}
Note that $T_0 + T_1$, i.e., the time of convergence to the set $D_M$, depends upon the initial condition $x(0)$ and the control bound $u_M$. Now we can state the main result of the paper.

\begin{Theorem}\label{main Thm}
Under the effect of control input \eqref{u t act}, the closed-loop trajectories of \eqref{cont affin sys} reach the set $S$ in a finite time $\bar T \leq T+T_0(x_0)+T_1$ for all $x_0\in \bar D\bigcap D$, where $T$ is the user-defined time as in Problem \ref{P1}.
\end{Theorem}

\begin{proof}
It follows from Theorem \ref{T conv origin} that for $x(0)\in D_M$, the closed-loop trajectories reach the set $S$ in the prescribed time $T$. From Theorem \ref{finite time set reach}, we obtain that for $x(0)\notin D_M$, the closed-loop trajectories reach the set $D_0$ in time $T_0(x_0)$ and from $D_0$, the closed-loop trajectories reach the set $D_M$ in time $T_1$. Hence, we have that for all $t\geq \bar T$, where $\bar T \leq T+T_0(x_0)+T_1$, $x(t)\in S$ for all $x(0)\in \bar D\bigcap D$. 
\end{proof}

\begin{Remark} \label{rem7}
Note that constant $k = \frac{l(1+\mu)}{k_1}   + \mu k_2$ in the definition of set $D_M$ decreases as the tuning parameter $\mu$ decreases. Hence, for smaller values of $\mu$, the set $D_M$ is larger. On the other hand, the parameter $\gamma$ in the definition of set $\bar D$ is such that $ \gamma \propto \frac{\mu}{{l(1+\mu)}/{k_1} + \mu k_2}$. So, $\gamma$ decreases with $\mu$, and saturates to the maximum value $\gamma_\infty = \frac{1}{{l}/{k_1} + k_2}$ as $\mu$ increases. Hence, the set $\bar D$ shrinks as $\mu$ decreases and saturates as $\mu$ grows larger. On the basis of these trade-offs, the user can choose appropriate value of $\mu$, depending upon requirements on the sets $D_M$ and $\bar D$ for a given problem.
\end{Remark}

\section{Numerical Example}\label{sec: simulation}
We consider the following system:
\begin{align*}
    \dot x_1 &= x_2 + x_1(x_1^2+x_2^2-1)+x_1u, \\
    \dot x_2 &= -x_1+\zeta(x_2)(x_1^2+x_2^2-1)+x_2u,
\end{align*}
where $\zeta(x) = (0.8+0.2e^{-100|x|})\tanh(x)$ and the set $S$ defined as $S = \{x\; |\; \|x\|\leq 1\}$. Note that in the absence of the control input, the trajectories diverge away from the set $S$, i.e., the set $S$ is unstable for the open-loop system. We choose $u_M = 7$, $\mu=0.05$ and use $V = \frac{1}{2}(x_1^2+x_2^2-1)^2$.
In addition, we set $T = 5$ seconds for the time transformation function in \eqref{P1_TT}.
With $\|x\|_S = x_1^2+x_2^2-r^2$, Assumptions \ref{assum: l k1 k2}, \ref{assum W c} are satisfied. 
Figures \ref{fig:x traj} and \ref{fig:u ub} respectively show the closed-loop system trajectories and the control input for different initial conditions. In Figure \ref{fig:x traj}, the set $D_0$ is denoted with red dashed-line and set $D_M$ is denoted with black dotted-line. 
Figure \ref{fig: energy} shows the variation of control Lyapunov function $V$ with time. The control Lyapunov function drops to $1e^{-15}$ within 5 sec for the case when $x(0)\in D_M$ and within 5.5 sec for the other cases. One can see from these figures that the system trajectories reach the set $S$ with bounded control input $u(t) \in [-u_M,\; u_M]$ within the prescribed finite time $T=5$ seconds when $x(0) \in D_M$ and within a finite time greater than 5 sec when $x(0)\notin D_M$. 

\begin{figure}[!ht]
    \centering
        \includegraphics[ width=1\columnwidth,clip]{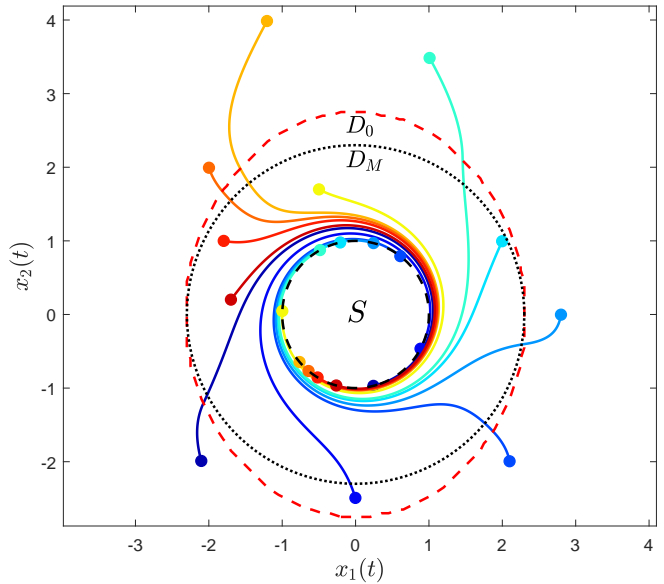}
        \caption{Closed-loop system trajectories for different initial conditions.}
    \label{fig:x traj}
\end{figure}

\begin{figure}[!ht]
    \centering
        \includegraphics[ width=1\columnwidth, trim={.95cm .15cm 1.5cm 1cm},clip]{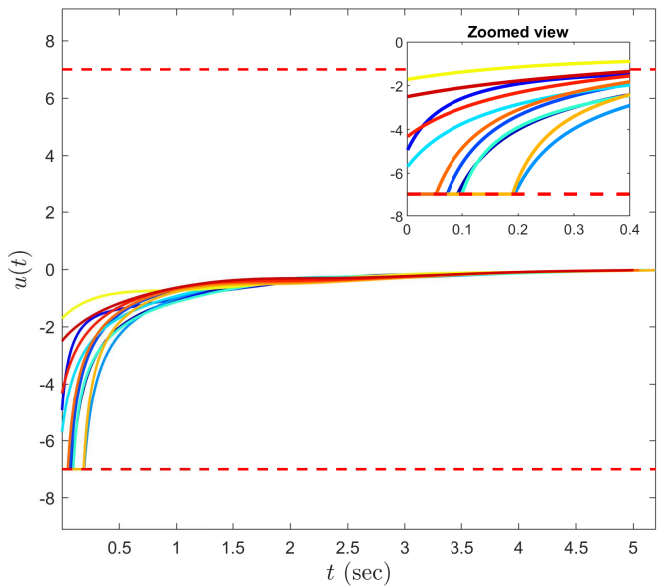}
    \caption{Control input signal $u(t)$ for the system trajectories in Figure \ref{fig:x traj}.}\label{fig:u ub}
\end{figure}

We use $\log$-scale to plot $V(t)$ with time, so that the variation of the function is clear when the values are very small. Also, note that unlike asymptotic convergence, which is plotted as a straight line on the $\log$-scale, the control Lyapunov functions drops super-linearly from $1e^{-3}$ to $1e^{-15}$. This verifies the finite-time convergence nature of the proposed controller.

\begin{figure}[!ht]
    \centering
        \includegraphics[ width=1\columnwidth,clip]{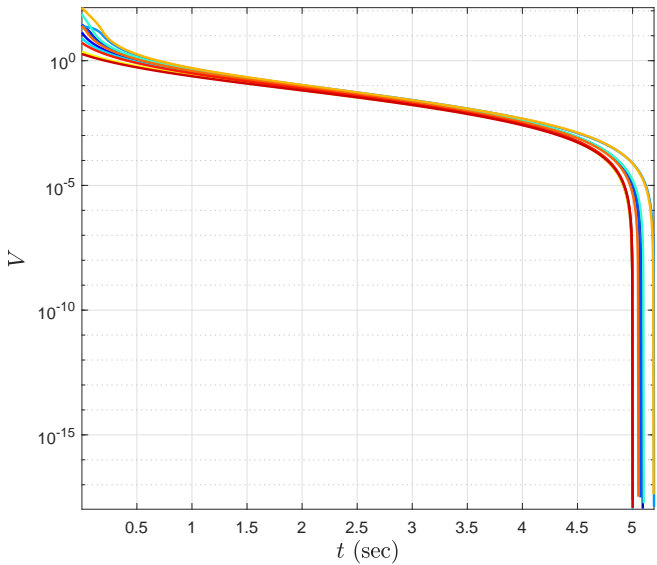}
    \caption{The evolution of the control Lyapunov function $V(t)$ with time for the system trajectories in Figure \ref{fig:x traj}.}\label{fig: energy}
\end{figure}

\begin{figure}[!ht]
    \centering
        \includegraphics[ width=1\columnwidth,clip]{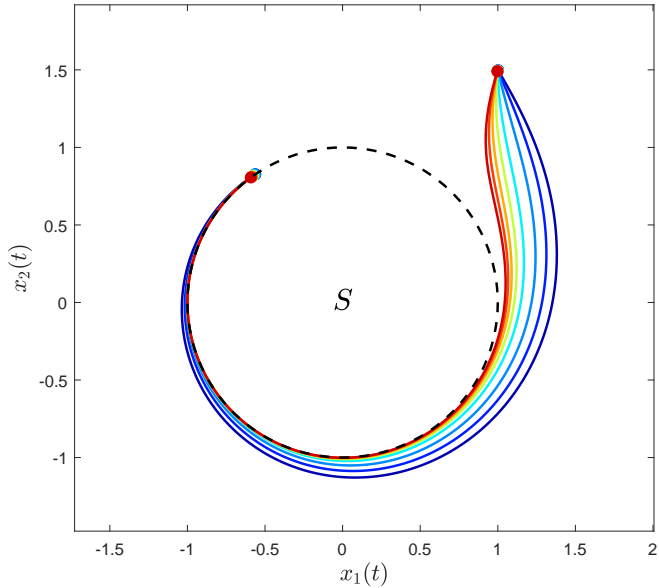}
        \caption{Closed-loop system trajectories for different values of $\mu \in [0.05, 0.35]$ (blue to red).}
    \label{fig:x2 traj}
\end{figure}

To illustrate the effect of the design parameter $\mu$, next we choose $x(0) = [1\; 1.5]^T$ and change $\mu$ from $0.05$ to  $0.35$.
Figures \ref{fig:x2 traj} and \ref{fig:u2 ub} respectively show the closed-loop system trajectories and the control input for different values of $\mu$. 
Figure \ref{fig: energy2} shows the variation of control Lyapunov function $V$ with time. Once again, the system trajectories reach the set $S$ within the prescribed finite-time $T=5$ seconds. As discussed in Remark \ref{rem7}, one can see from Figure \ref{fig:u2 ub} that a smaller value of $\mu$ results in less control effort in the beginning. Also, from \eqref{v prime bound}, it is evident that smaller value of $\mu$ would lead to a slower convergence-rate, which is also clear from Figure \ref{fig:x2 traj}.

\begin{figure}[!ht]
    \centering
        \includegraphics[ width=1\columnwidth,clip]{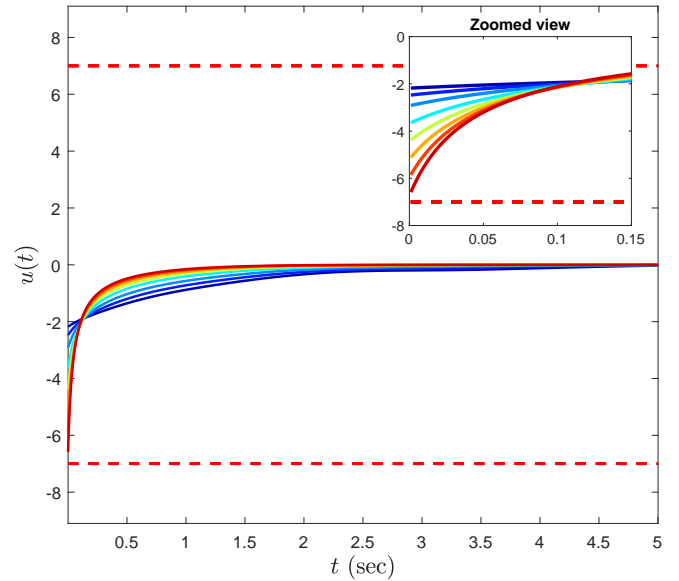}
    \caption{Control input signal $u(t)$  for the system trajectories in Figure \ref{fig:x2 traj}.}\label{fig:u2 ub}
\end{figure}

\begin{figure}[!ht]
    \centering
        \includegraphics[ width=1\columnwidth,clip]{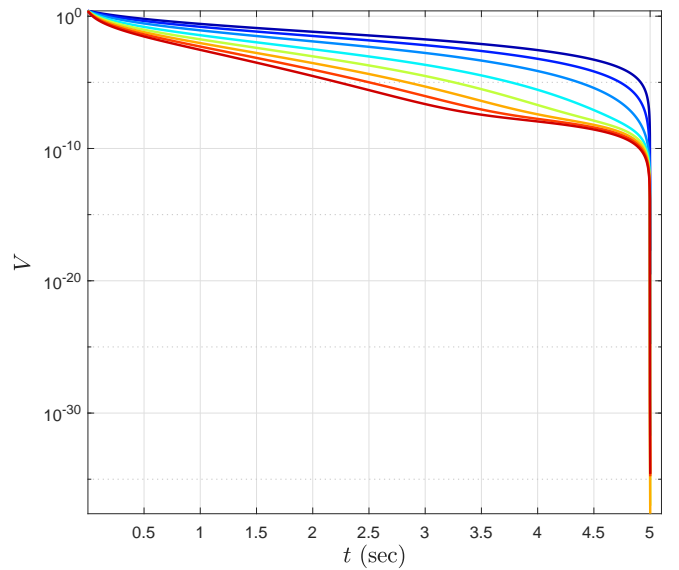}
    \caption{The evolution of the control Lyapunov function $V(t)$ with time for the system trajectories in Figure \ref{fig:x2 traj}.}\label{fig: energy2}
\end{figure}



\section{Conclusions}\label{sec: conclusion}
We presented a method of designing control law for a class of control-affine systems with input constraints so that the closed-loop trajectories reach a given set in a prescribed time. We showed that the set of initial conditions $D_M$ from which this convergence is guaranteed, is a function of the convergence time, and that this set grows larger as the input bound or the time of convergence increases. Furthermore, the proposed controller drives the system trajectories to this set within a finite time, that depends upon the initial condition. Finally, we introduced a new set of sufficient conditions to ensure the existence of a prescribed-time stabilizing controller. 

\bibliographystyle{IEEEtran}
\bibliography{myreferences}

\appendices

\section{Proof of Lemma \ref{lemma u bound}}\label{App proof lemma 1}

\begin{proof}
It follows from \eqref{v prime bound} and \eqref{w bound} that
\bequ
V'(x_s) \leq -c \mu \|x_s\|_S^2 \leq -\mu \frac{c}{c_2} V(x_s), \label{vprime2}
\eequ
where \eqref{v bound} and \eqref{c2 bound} are used in the last inequality. Per Comparison lemma, the function $V(x_s(s))$ satisfies
\bequ
V \leq V_0 e^{-\mu \frac{c}{c_2}s}, \label{vsol}
\eequ
with $V_0 = V(x_s(0))$. Using  \eqref{c1 c2 bound} and \eqref{vsol}  in \eqref{v bound} yields 
\bequ
\|x_s\|_S \leq m_1 e^{-m_2 s}. \label{xbound}
\eequ 
with $m_1 = \sqrt{\frac{V_0}{c_1}} $ and $m_2 = \frac{\mu c}{2 c_2}>0$. Now, using \eqref{a b a0 b0 rel} for the case when $b_0 \neq 0$, the control input \eqref{u s} can be rewritten as 
\bequ
\bar u = -\frac{a_0}{b_0} - \mu\frac{|b_0|}{b_0} \sqrt{\Big(\frac{a_0}{b_0 \theta'(s)}\Big)^2+\Big(\frac{b_0}{\theta'(s)}\Big)^2},
\eequ
The control input upper bound can be written as
\begin{align*}
    |\bar u| & \leq  \frac{|a_0|}{|b_0|} \Big(1+\frac{\mu}{\theta'(s)}\Big) + \mu\frac{|b_0|}{\theta'(s)} \\
    & \overset{\eqref{a0 b0 bounds}}{\leq }\bigg(\frac{l(\theta'(s)+ \mu)}{k_1}   + \mu{k_2}\bigg) \frac{\|x_s\|_S}{\theta'(s)}.
\end{align*}
For the considered candidate time transformation function in \eqref{P1_TT}, one can further obtain
\begin{align}\label{u bound2}
    |\bar u| \leq  \bigg(\frac{l(1+\mu)}{k_1}   + \mu k_2\bigg) \frac{\|x_s\|_S}{e^{-s/T}} \overset{\eqref{v bound}}{\leq} \frac{k}{\sqrt{c_1}}\frac{\sqrt V}{e^{-s/T}},
\end{align}
where $k =  \big(\frac{l(1+\mu)}{k_1}   + \mu k_2\big)$. Using \eqref{vprime2}, we obtain
\begin{align}\label{u bound3}
    |\bar u| \leq  \frac{k}{\sqrt{c_1}} \sqrt {V_0}e^{-\mu\frac{c}{2c_2}s+s/T}.
\end{align}
Hence, for $T>\frac{2c_2}{\mu c}$, we obtain that $\lim_{s\to\infty}|\bar u|= 0$.
\end{proof}

\section{Proof of Lemma \ref{lemmma u cont}}\label{App proof lemma 2}
\begin{proof}
It is clear from \eqref{u t act} that $|u(t)|\leq u_M$ for all $t\geq 0$. From \eqref{u bound3}, we know that for any $\epsilon>0$, there exists a $\delta>0$ such that if $\|x\|_S<\delta$, then $|\bar u|<\epsilon \triangleq m_3 \delta^{m_4}$. Also, for $\delta<\Big(\frac{u_M}{m_3}\Big)^\frac{1}{m_4}$, we know that $|\bar u|<u_M$, which implies that $V'(x)<0$. So, $\bar u$ satisfies the \textit{small-gain} property \cite{sontag1989universal} and hence, it is continuous at the origin. Since functions $a,b$ are continuous and due to \eqref{assum: l k1 k2}, $|b|>0$ for $\|x\|_S>0$, we have that $\bar u$ is continuous for all $x$. 

Now, assume that $x(0)\notin D_0$ and $t = T_0$ is the instant when the trajectories of \eqref{cont affin sys} enter the set $D_0$, i.e., $x(T_0)\in \textrm{bd}(D_0)$ and $x(t)\notin D_0$ for all $t<T_0$. Then, control input at time instant $T_0^-$ is given by
\begin{align*}
    u(T_0^-) & = u_M\frac{\bar u(1,x(T_0))}{|\bar u(1,x(T_0))|} = \bar u(1,x(T_0)),
\end{align*}
since $x(T_0)\in \textrm{bd}(D_0)$ implies $|\bar u| = u_M$. At time instant $T_0^+$, we have $u(T_0^+) = \bar u(1,x(T_0))$. This implies $u(T_0^-) = u(T_0^+)$. Now, let $t = T_M$ be the time instant when the trajectories of \eqref{cont affin sys} reach the set $D_M$. From \eqref{u t act}, we know that 
\begin{align*}
    u(T_M^+) &=  \bar u(\theta'(\theta^{-1}(t-T_M)),x(t))|_{t = T_M} \\
    & = \bar u(\theta'(\theta^{-1}(0)),x(T_M))\overset{\eqref{theta prop}}{=} u(1,x(T_M)) = u(T_M^-),
\end{align*}
which completes the proof. 
\end{proof}

\end{document}